\documentclass[11pt,a4paper]{amsart}

\usepackage[english]{babel}

\usepackage[top=2cm,bottom=2cm,left=3cm,right=3cm,marginparwidth=1.75cm]{geometry}
\usepackage{subcaption}
\usepackage{tabularx}
\usepackage{amsmath,amssymb,bm,color}
\usepackage[english]{babel}
\usepackage{mathtools}
\usepackage{enumerate}
\usepackage[numbers]{natbib}
\usepackage[noend]{algpseudocode}
  
\usepackage{multirow}
\usepackage{xcolor}
\usepackage{float}
\usepackage[section]{placeins}

\usepackage{pgfplots}
\pgfplotsset{compat=newest}
\usepackage{tikz}
\usetikzlibrary{arrows,matrix,positioning,calc}
\usepackage[utf8]{inputenc}

\usepackage[T1]{fontenc}
\usepackage{epsfig}

\usepackage{mathdots}
\usepackage{xspace}

\usepackage[linesnumbered,ruled,vlined,titlenumbered]{algorithm2e}
\usepackage{color}
\usepackage{booktabs}
\usepackage{verbatim}
\usepackage[hidelinks]{hyperref}

\usepackage{lipsum}
\usepackage{enumitem}
\usepackage{colortbl}

\newtheorem{theorem}{Theorem}
\newtheorem{problem}{Problem}

\newtheorem{proposition}[theorem]{Proposition}

\theoremstyle{definition}

\newtheorem{definition}[theorem]{Definition}
\theoremstyle{remark}

\usepackage[final,tracking=true,kerning=true,spacing=true,factor=1100,stretch=10,shrink=20]{microtype}

\renewcommand{\vec}[1]{\ensuremath{\mathbf{#1}}}
  
\renewcommand{\S}{\ensuremath{\vec{S}}}
\renewcommand{\H}{\ensuremath{\vec{H}}}
\newcommand{\0}{\ensuremath{\vec{0}}}
\newcommand{\e}{\ensuremath{\vec{e}}}

\newcommand{\Fq}{\ensuremath{\mathbb{F}_q}}

\DeclareMathOperator{\rank}{rk}
\DeclareMathOperator{\wtH}{wt}

\DeclareMathOperator{\supp}{supp}

\renewcommand{\c}{\vec{c}}
\renewcommand{\e}{\vec{e}}

\renewcommand{\r}{\vec{r}}
\newcommand{\s}{\vec{s}}

\newcommand{\x}{\vec{x}}

\newcommand{\C}{\vec{C}}

 \newcommand{\E}{\vec{E}}
\newcommand{\R}{\vec{R}}
\newcommand{\X}{\vec{X}}

\newcommand{\G}{\vec{G}}

\newcommand{\M}{\vec{M}}

\newcommand{\bbN}{\mathbb{N}}

\newcommand{\cC}{\mathcal{C}}

\newcommand{\cI}{\mathcal{I}}
\newcommand{\cJ}{\mathcal{J}}
\newcommand{\cL}{\mathcal{L}}

\newcommand{\cU}{\mathcal{U}}

\newcommand{\cT}{\mathcal{T}}

\newcommand{\GsupCode}{\G'}
\newcommand{\HsupCode}{\H'}

\newcommand{\Cl}{\cC_ \ell}

\newcommand{\gb}{\genfrac{[}{]}{0pt}{}}
\title{Interleaved Prange: A New Generic Decoder for Interleaved Codes}
\begin{document}
 
	\author[A. Porwal]{Anmoal Porwal}
	\address{Department of Electrical and Computer Engineering\\
		Technical University of Munich\\
 Germany
	}
	\email{anmoal.porwal@tum.de}
	\author[L. Holzbaur]{Lukas Holzbaur}
	\address{Department of Electrical and Computer Engineering\\
		Technical University of Munich\\
 Germany
	}
	\email{lukas.holzbaur@tum.de}
	\author[H. Liu]{Hedongliang Liu}
	\address{Department of Electrical and Computer Engineering\\
		Technical University of Munich\\
 Germany
	}
	\email{lia.liu@tum.de}
	\author[J. Renner]{Julian Renner}
	\address{Department of Electrical and Computer Engineering\\
		Technical University of Munich\\
 Germany
	}
	\email{julian.renner@tum.de}
	\author[A. Wachter-Zeh]{Antonia Wachter-Zeh}
	\address{Department of Electrical and Computer Engineering\\
		Technical University of Munich\\
 Germany
	}
	\email{antonia.wachter-zeh@tum.de}
 
	\author[V. Weger]{Violetta Weger}
	\address{Department of Electrical and Computer Engineering\\
		Technical University of Munich\\
 Germany
	}
	\email{violetta.weger@tum.de}
 \keywords{Information Set Decoding, Interleaved Codes, Code-Based Cryptography}
\maketitle

\begin{abstract}
Due to the recent challenges in post-quantum cryptography, several new approaches for code-based cryptography have been proposed.
For example, a variant of the McEliece cryptosystem based on interleaved codes was proposed. In order to deem such new settings secure, we first need to understand and analyze
the complexity of the underlying problem, in this case the problem of decoding a random interleaved code. A simple approach to decode such codes, would be to randomly choose a vector in the row span of the received matrix and run a classical information set decoding algorithm on this erroneous codeword.
In this paper, we propose a new generic decoder for interleaved codes, which is an adaption of the classical idea of information set decoding by Prange and perfectly fits the interleaved setting.  We then analyze the cost of the new algorithm and a comparison to the simple approach described above shows the superiority of Interleaved Prange.    
\end{abstract}

\section{Introduction}\label{sec:intro}
Code-based cryptography is one of the most promising and prominent candidates for post-quantum cryptography, which is reflected in the NIST standardization process \cite{NIST}. Although the third round of submissions has already been completed, and the classical McEliece system \cite{classicMC} has been chosen as a finalist, there are still many open challenges in the area. For example the lack of efficient and secure signature schemes \cite{round3}, but also the compelling task of reducing the key sizes of the original McEliece system is still open.
For this reason, researchers have proposed several alternatives to the classical scheme of McEliece, not only by changing the underlying code family, but also by considering different settings, for example by employing the rank metric \cite{rollo,rqc}, the Lee metric \cite{thomas, leeZ4, leenp} or by using interleaved codes. 
The latter approach has been proposed in \cite{IG,holzbaur2019decoding,IR}.
The simple reasoning behind this proposal is that an interleaved code has a larger error-correction capability than a non-interleaved code.

A codeword of an $ \ell$-interleaved code is an $ \ell\times n$ matrix over $\Fq$, where each row is a codeword of a constituent linear code of blocklength $n$ over $\Fq$. 
In this work, we consider the decoding problem for \emph{homogeneous} interleaved codes, where the same constituent code is used for all the rows.

An interleaved code $\Cl$ is especially well-suited for channels that are prone to burst errors, where $t$ burst errors can be modeled as the addition of an $ \ell\times n$ matrix~$\E$ with $t$ non-zero columns to a codeword in $\Cl$. We say that $\E$ has column weight $t$.

A generic decoder for any linear interleaved code was proposed in \cite{haslach1999decoding,metzner1990general}. When the interleaving order $ \ell$ is at least the number of column errors $t$, this decoder guarantees to correct (efficiently) any full-rank error of weight up to $d-2$, where $d$ is the minimum distance of the constituent code. This decoder was generalized in \cite{haslach2000efficient,roth2014coding} for the case $ \ell<t$ and guarantees to decode any error $\E$ of weight~$t$ if $2t-\rank(\E)\leq d-2$. 
However, there is no known efficient decoder for interleaved codes with an arbitrary constituent code when $ \ell \ll t$. 
In fact, it can be shown that the corresponding decisional problem, called Interleaved Decoding (ID) problem,  is at least as hard as the decisional Syndrome Decoding (SD) problem.

This  fact implies that interleaved codes  are a well-suited alternative for code-based cryptography.
It is therefore of interest to understand and analyze the complexity of decoding a generic interleaved code not only from a coding-theoretic perspective, but also in order to assess
the security of code-based cryptosystems based on interleaved codes.

In this paper, we consider algorithms for the ID problem when $\ell \ll t < d$ for arbitrary linear constituent codes. 
We can categorize the generic decoding algorithms for interleaved codes into three types:
\begin{enumerate} 
\item Algorithms that reduce the problem to the classical SD problem.
\item Algorithms that reduce the problem to a low-weight codeword finding (CF) problem.
\item Algorithms that do not reduce the problem to either CF or SD. We present one such novel algorithm inspired by Prange's information set decoder \cite{prange}.
\end{enumerate}

We remark that we will be content with finding just a subset of the $t$ error positions since then the problem reduces to a much easier problem as  the complexity is exponential in $t$.
For the third family of algorithms, we propose Interleaved Prange. The classical Prange algorithm \cite{prange} can be described as picking $k$ columns of the generator matrix $\G$, where the algorithm is successful if the corresponding positions are error-free, i.e., their complement of $n-k$ positions contains the support of the error. Alternatively one can pick $k+1$ columns from $\left[\begin{smallmatrix} \G \\ \r \end{smallmatrix}\right]$ where $\r$ is the received word and check whether the $(k+1) \times (k+1)$ submatrix formed by these columns is rank-deficient.
This can be generalized to interleaved codes, which is the main idea of our algorithm Interleaved Prange:  we pick $k+\ell$ columns of $\left[\begin{smallmatrix} \G \\ \R \end{smallmatrix}\right]$ where $\R$ is an $\ell \times n$ matrix containing the $\ell$ received words as rows, and check if the rank of the $(k+\ell) \times (k+\ell)$ submatrix formed by these columns is less than $k+\ell$.
The main contribution of this paper is the proposal and the analysis of the new decoding algorithm Interleaved Prange. 

This paper is structured as follows. In Section \ref{sec:prelim} we introduce the notation and results for interleaved codes which are essential for the remainder of the paper. We present the three types of interleaved decoding algorithms in Section \ref{sec:decode} together with their corresponding complexity analysis. A comparison of their asymptotic cost, given in Section \ref{sec:comp}, then shows that the newly proposed algorithm, Interleaved Prange, outperforms the straight-forward decoders.  
Finally, we conclude this paper in Section \ref{sec:concl}.

\section{Preliminaries}\label{sec:prelim}

Let us first introduce the notation that is used throughout this paper. 
For a prime power $q$, let us denote by $\mathbb{F}_q$ be a finite field with $q$ elements. 
We denote matrices and vectors by bold capital, respectively lower case letters. For $k \leq n$ positive integers and a matrix $\G \in \mathbb{F}_q^{k \times n}$ we denote by $\langle \G \rangle$ its rowspan, by $\G^\top$ the transposed matrix and by $\text{rk}(\G)$ its rank. For a vector $\x \in \mathbb{F}_q^n$, we  will denote by $\text{wt}(\x)$ the Hamming weight of $\x$, that is the size of its support $\text{supp}(\x)$. For a matrix $\X \in \mathbb{F}_q^{k \times n}$ we will denote by $\text{wt}(\X)$ the number of non-zero columns of $\X.$  For a set $\mathcal{S}$ we will denote by $|\mathcal{S} |$ its cardinality. The set of all integers between 1 and $n$ is denoted by $[1,n].$ Finally, for a set $\mathcal{I} \subseteq [1,n]$  of size $r$ and a matrix $\G \in \mathbb{F}_q^{k \times n}$, we denote by $\G_{\mathcal{I}} \in \mathbb{F}_q^{k \times r}$ the matrix consisting of all columns of $\G$ indexed by $\mathcal{I}.$ For a vector $\x \in \mathbb{F}_q^n$, we denote by  $\supp(\x)$ its support, that is the indices of the non-zero entries of $\x$. Similarly for  a matrix $\X \in \mathbb{F}_q^{k \times n}$ we denote by $\supp(\X)$ the indices of the non-zero columns.\\
 
 A linear subspace $\mathcal{C} \subseteq \mathbb{F}_q^n$ of dimension $k$ is called a linear code of length $n$ and dimension $k$. We call this an $[n,k]_q$ code of rate $R=\frac{k}{n}.$
 For a linear code $\mathcal{C} \subseteq \mathbb{F}_q^n$ we can also define its minimum distance to be 
 $$d(\mathcal{C}) = \min\{ \text{wt}(\c) \mid c \in \mathcal{C}, \c \neq 0 \}.$$

An $[n,k]_q$ linear code can be represented either through a generator matrix $\G \in \mathbb{F}_q^{k \times n}$, which has the code as image, or through a parity-check matrix $\H \in \mathbb{F}_q^{(n-k) \times n}$, which has the code as right kernel. For any $\x \in \mathbb{F}_q^n$, we call $\s = \x\H^\top \in \mathbb{F}_q^{n-k}$ the syndrome of $\x.$

It is well known that random codes of large blocklength over $\mathbb{F}_q$ achieve with high probability the minimum distance given by the Gilbert-Varshamov bound, that is
$$\delta=\frac{d(n)}{n} = H_q^{-1}(1-R),$$ where we denote by $H_q$ the $q$-ary entropy function.

\begin{definition}
Let $\mathcal{C} \subseteq \mathbb{F}_q^n$ be a linear code of dimension $k$ with generator matrix $\G \in \mathbb{F}_q^{k \times n}.$
 The homogeneous interleaved code of interleaving order $\ell$ of $\mathcal{C}$ is defined as 
 $$\mathcal{C}_\ell = \{ \C \in \mathbb{F}_q^{\ell \times n} \mid \C=\M\G, \M \in \mathbb{F}_q^{\ell \times k} \}. $$
\end{definition}
Thus, the codewords of an interleaved code are $\ell \times n$ matrices. Let $\H$ be a parity-check matrix of $\mathcal{C}$ and consider the interleaved code $\mathcal{C}_\ell$. The syndrome of $\X \in \mathbb{F}_q^{\ell \times n}$ is then given by $$\S = \X\H \in \mathbb{F}_q^{\ell \times (n-k)}.$$
Decoding an interleaved code with an arbitrary constituent code can be seen as the following problem.
 \begin{problem}[Interleaved Syndrome Decoding (ISD) Problem]\label{prob:ISD} Let $\ell\geq2$ be a positive integer.
   Given $\H\in\Fq^{(n-k)\times n}$, $\S\in\Fq^{ \ell\times (n-k)}$, and $t\in\bbN$, decide if there exists a matrix $\E\in\Fq^{ \ell\times n}$ of weight at most $t$, such that $\H\E^\top=\S^\top$.
 \end{problem}
 This problem is  equivalent to the Interleaved Decoding (ID) problem.
\begin{problem}[Interleaved Decoding (ID) Problem]\label{prob:ID}
  Given $\G\in\Fq^{k\times n}$, $\R\in\Fq^{ \ell\times n}$, and $t\in\bbN$, decide if there exists a matrix $\E\in\Fq^{ \ell\times n}$ of column weight at most $t$, such that each row of $\R- \E$ is in $\langle \G \rangle$.
\end{problem}
This problem can be shown to be NP-hard by a reduction from the Hamming-metric SD problem, which has been proven to be NP-complete in \cite{np,barg}.

\begin{problem}[Hamming Syndrome Decoding (SD) Problem]\label{prob:SD}
  Given $\H\in\Fq^{(n-k)\times n}$, $\s\in\Fq^{n-k}$, and $t\in\bbN$, decide if there exists a  $\e\in\Fq^{n}$ of  weight at most $t$, such that $\s=\e\H^\top$.
\end{problem}

\begin{theorem}
The Interleaved Syndrome Decoding Problem (Problem \ref{prob:ISD}) is NP-complete.
\end{theorem}

\begin{proof}
We show the NP-hardness of Problem \ref{prob:ISD} by a reduction from the classical Hamming SD. For this, take a random instance $\H \in \mathbb{F}_q^{(n-k) \times n}, \s \in \mathbb{F}_q^{n-k}$ and $t\in \mathbb{N}$ of the Hamming SD. Now define $\S= \begin{pmatrix} \s \\ \vdots \\ \s \end{pmatrix} \in \mathbb{F}_q^{\ell \times (n-k)}.$ Assume we have an oracle for Problem \ref{prob:ISD}.
\begin{itemize}
\item If the answer is `yes' on the input $\H, \S,t$, then this is also the correct answer to the Hamming SD. In fact, if there exists $\E \in \mathbb{F}_q^{\ell \times n}$, such that $\H\E^\top=\S^\top$ and at most $t$ columns of $\E$ are  non-zero, then any column, e.g., the first column $\e$, of $\E$ is a solution to the Hamming SD, as $\H\e=\s$ and $\wtH(\e)\leq t.$
\item If the oracle returns `no' on the input $\H, \S,t$, then this is also the correct answer to the Hamming SD. In fact, if there was a solution $\e$ to the Hamming SD then $\E= \begin{pmatrix} \e \\ \vdots \\ \e \end{pmatrix}$ would have been a solution to the interleaved SD.
\end{itemize}
Finally, we remark that for any candidate $\E$ we can  check in polynomial time, whether $\E$ is a solution to the interleaved SD. Thus, the problem is also in NP.  
\end{proof}

\section{Decoding Algorithms}\label{sec:decode}
In this section we present three types of generic decoding algorithms for interleaved codes. 
That is, given  $\G\in\Fq^{k\times n}$, $\R\in\Fq^{ \ell\times n}$, and $t\in\bbN$, these algorithms find a matrix $\E\in\Fq^{ \ell\times n}$ of column weight at most $t$, such that each row of $\R- \E$ is in $\langle \G \rangle$.

In the following, we  assume that $\G \in \mathbb{F}_q^{k \times n}$ and
a set of error positions $\cT \subseteq [1, n]$ of size $t$ is chosen uniformly at random.
Then ones takes a $\ell \times n$ zero matrix $\E$ and sets each column at these $t$ error positions equal to a random vector in $\Fq^\ell$.
Thus $\E_\cT$ is a random matrix in $\Fq^{\ell \times t}$, and $\E$ is a random matrix in $\mathbb{F}_q^{ \ell \times n}$ of column weight at most $t$.
Finally, we choose $\M \in \mathbb{F}_q^{ \ell \times k}$ uniformly at random and compute the received matrix
$\R = \C +\E$ where $\C = \M\G$. Thus, we  assume that at least one solution to the ID problem exists.
For interleaved cryptosystems, $t$ is typically close to the minimum distance of $\G$ which we denote by $d$.

\subsection{SD-based Algorithms}\label{sec:SD}
The most straightforward way to solve the ID problem is to simply pick a random non-zero vector $\r$ in the rowspan of $\R$ and solve the resulting SD problem with the parity-check matrix $\H \in \Fq^{(n-k) \times n}$ of the constituent code and the syndrome $\s=\r\H^\top \in \mathbb{F}_q^{n-k}$. Since information set decoding (ISD) attacks are the best known algorithms
to solve the SD problem, we call this \emph{Random \textlangle ISD\textrangle} (where \textlangle ISD\textrangle\ can be any ISD algorithm such as Prange, Stern \cite{stern}, etc.).
 
We assume that it is enough to recover only a part of the non-zero columns of $\E$ since this knowledge will reduce the problem already to a much easier problem.
 \begin{algorithm}
   \caption{Random ISD}\label{algo:randomisd}
   \SetAlgoLined
   \DontPrintSemicolon
 
   \KwIn{A generator matrix $\G \in\Fq^{k \times n}$of $\mathcal{C}$ and a received matrix $\R = \C + \E \in \Fq^{\ell \times n}$ where $\text{wt}(\E)=t$} 
   \KwOut{A nonempty subset $\cU \subseteq \supp(\E)$}
  Compute a parity-check matrix $\H \in \mathbb{F}_q^{(n-k) \times n}$ of $\mathcal{C}$ \;
   Pick a non-zero $\r \in \langle \R\rangle$ at random\;
   Compute $\s = \r \H^\top$\;
   Use an ISD algorithm that can find errors $\e$ of any weight belonging to some fixed subset of $[1, t]$ and run it with inputs $\H,\s$\;
   When the ISD algorithm outputs an error $\e$, return $\supp(\e)$
 \end{algorithm}
If the success probability of the employed ISD algorithm of finding an error of weight $v$ is denoted by $P(v)$, then the success probability of the Random \textlangle ISD\textrangle\ approach is given by 
$$
\sum_{v=0}^{t}
    \dfrac{\binom{t}{v} (q-1)^{v}}{q^t}
    \cdot
    P(v),
$$
Note that $P(v)$ is simply zero for all those error weights $v$ which the chosen ISD algorithm is not designed to solve for. Here $\binom{t}{v}\frac{(q-1)^v}{q^t}$ denotes the probability that the chosen $\r$ has an error $\e$ of weight $v$. In fact, by choosing a random codeword $\r \in \langle \R\rangle$, this results in an error vector $\e$ which is a random linear combination of the rows of $\E \in \mathbb{F}_q^{\ell \times n}$ and thus when $\e$ is restricted to the $t$ error positions it looks like a vector drawn uniformly at random from $\Fq^{t}$.
Note that this approach comes with a failure probability as the errors generally have weight greater than the unique decoding radius of $\G$. However, this probability is negligible as the error weights are less than the minimum distance of $\G$.

For the complexity analysis, let us consider first that we employ the ISD algorithm of Prange \cite{prange}.
This algorithm has a success probability of $$P(v) = \binom{n-k}{v}\binom{n}{v}^{-1}.$$
Hence the success probability of Random Prange is given by
$$\sum_{v=0}^t \frac{\binom{t}{v}(q-1)^v}{q^t} \binom{n-k}{v}\binom{n}{v}^{-1}.$$
To get an upper bound on the asymptotic complexity of Random Prange, we can give a lower bound on the success probability, e.g., by considering just the term in the summation where $v= t\frac{q-1}{q}$ (a reasonable choice since this is the most likely error weight in the chosen $\r$, i.e., this $v$ maximizes $\binom{t}{v}\frac{(q-1)^v}{q^t}$).

In order to give an asymptotic complexity, we first consider the parameters $k,t$ as functions in $n$ and define 
\begin{align*}
    R &= \lim_{n \to \infty} \frac{k(n)}{n},\\
    T &= \lim_{n \to \infty} \frac{t(n)}{n} = H_q^{-1}(1-R).
\end{align*}
To ease the notation, we also introduce the  asymptotics of the binomial coefficient, denoted by 
 \begin{align*} H(F, G) & := \lim_{n\to \infty} \frac{1}{n} \log_{q} \left(  \binom{f(n)}{g(n)} \right) \\ 
 & = F \log_{q}(F) - G \log_{q}(G) - (F-G) \log_{q} (F-G), \end{align*} 
 where $f(n), g(n)$ are integer-valued functions such that $\lim\limits_{n\to \infty} \frac{f(n)}{n} = F$ and $\lim\limits_{n \to \infty} \frac{g(n)}{n} = G$.

Thus, we get the following upper bound.
\begin{proposition}
The asymptotic complexity of Random Prange on an $\ell$\nobreakdash-interleaved random code over $\mathbb{F}_q$ with length $n$ and dimension $k$ is given by at most $q^{ne(R,q)}$, where 
\begin{align*} e(R,q) = H(1,T(q-1)/q)-H(1-R, T(q-1)/q). \end{align*}
\end{proposition}
If we employ Stern's ISD algorithm \cite{stern}, we get a slight improvement.
However, note that Stern's algorithm (at least in its conventional formulation) only solves the SD problem for a fixed error weight $w$.
If instead in each iteration we run Stern $t$ times for all error weights $w \in [1, t]$, this gives us a straightforward extension of the algorithm that works for all errors with weights in $[1, t]$.
While this of course increases the cost of one iteration, it turns out that asymptotically the cost remains the same and since this formulation can only improve the probability of success of Random Stern, we will consider this version.

The cost of Random Stern's algorithm is in $\mathcal{O}(I \cdot C)$, where $I$ is the expected number of iterations and $C$ the cost of one iteration. This is given by
\begin{align*}
    I &= \left( \sum_{v=0}^t \frac{\binom{t}{v}(q-1)^v}{q^t} \binom{(k+\ell'_v)/2}{w'_v/2}^2 \binom{n-k-\ell'_v}{v-w'_v} \binom{n}{v}^{-1}\right)^{-1} \\
    C &= \sum_{v=1}^{t} C_v \text{ where } C_v = \binom{(k+\ell'_v)/2}{w'_v/2}q^{w'_v/2} + \binom{(k+\ell'_v)/2}{w'_v/2}^2q^{w'_v-\ell'},
\end{align*}
where $0\leq w'_v \leq \min\{k+\ell', v\}, 0 \leq \ell'_v \leq n-k$ are the internal parameters of Stern's algorithm that can be optimized individually for each of the $t$ runs to give the lowest cost. 

To get an upper bound on the asymptotic complexity of Random Stern, we again just consider the $v_0=t\frac{q-1}{q}$ term in the summation in $I$'s formula.
For this let us consider additionally the parameters $w_{v_0}'$ and $\ell_{v_0}'$ as functions in $n$ and define
\begin{align*}
    W' &= \lim_{n \to \infty} \frac{w'_{v_0}(n)}{n},\\
    L' &= \lim_{n \to \infty} \frac{\ell'_{v_0}(n)}{n}.
\end{align*}

\begin{proposition}
The asymptotic complexity of Random Stern on an $\ell$\nobreakdash-interleaved random code over $\mathbb{F}_q$ with length $n$ and dimension $k$ is given by at most $q^{ne(R,q)}$, where 
\begin{align*} e(R,q) = &   H(1,T(q-1)/q)-2H((R+L')/2,W'/2) \\ & - H(1-R-L',T(q-1)/q-W')  \\ & + \max\{H((R+L')/2,W'/2) +W'/2,   \\ &  2H((R+L')/2,W'/2) +W'-L' \}.
\end{align*}
\end{proposition}

\subsection{CF-based Algorithms}\label{sec:CF}
A different approach is the following. 
Having received the matrix $\R$, note that the code generated by $\left[\begin{smallmatrix} \G \\ \R \end{smallmatrix}\right] \eqqcolon \GsupCode$ is the same as the code generated by  $\left[ \begin{smallmatrix} \G \\ \E \end{smallmatrix}\right]$.
Thus the problem reduces to finding a low-weight codeword in the code $\langle \GsupCode \rangle$ of dimension $k+ \ell$.
Let us denote by $\HsupCode \in \mathbb{F}_q^{(n-(k+ \ell)) \times n}$ a parity-check matrix of the code $\langle {\GsupCode} \rangle.$

 \begin{algorithm}
   \caption{CF-based Algorithm}\label{algo:cfbased}
   \SetAlgoLined
   \DontPrintSemicolon
   \KwIn{A generator matrix $\G \in\Fq^{k \times n}$ and a received matrix $\R = \C + \E \in \Fq^{\ell \times n}$ where $\langle \E \rangle$ has minimum distance $w$} 
   \KwOut{A nonempty subset $\cU \subseteq \supp(\E)$}
  Compute the parity-check matrix $\H' \in \Fq^{n-k-\ell}$ of the code $ \langle \G' \rangle$.\;
   Use a CF algorithm with inputs $\H'$ and $w$ to find a set $\cL$ of codewords of weight $w$ in $\langle \G' \rangle$ \label{step:CSD}\;
   Return $\cU=\cup_{\e\in\cL}\supp(\e)$
 \end{algorithm}

Algorithm \ref{algo:cfbased} gives a framework of finding the support of $\E$ from $\HsupCode$ by using a low-weight codeword finding algorithm (e.g., \cite[Algorithm 1]{otmani2011efficient}) as a subroutine.

The complexity of this approach is the same as the complexity of the CF algorithm used for finding low-weight codewords in the code $\langle \GsupCode \rangle$. For example one might employ the well-known ISD algorithm by Stern. 
However, this approach comes with a possibly large failure probability, as $\langle \GsupCode \rangle$ might contain many low-weight codewords whose support is not a subset of the original $t$ error positions.

One reason to suspect this is as follows. Recall that a random $[n,k]$ linear code over $\mathbb{F}_q$ with minimum distance $d$ has on average
$q^{k-n+w} \binom{n}{w}$ many codewords of weight $d\leq w \leq n.$
If we treat $\G'$ as a random matrix then it has approximately $q^{k+\ell-n+w} \binom{n}{w}$ many codewords of weight $w$, while only $q^{\ell-t+w}\binom{t}{w}$ many of those are from $\langle \E_\cT \rangle$.
This approach gives a failure probability of 
$ 1-q^{-k-\ell-t}\binom{t}{w}\binom{n}{w}^{-1}$.

However, this is imprecise as $\G'$ is not entirely random, but such that the last $\ell$ rows have only $t$ non-zero columns.
Unfortunately, an accurate analysis of the failure probability for these algorithms is complicated, but the above computation does gives evidence that it is quite large.
As the other algorithms have a failure probability that is either negligible or at least allow for a more tractable analysis, it is hard to compare them to CF-based algorithms.

\subsection{Novel approach: Interleaved Prange}\label{sec:IP}
We propose a new algorithm (Algorithm \ref{algo:IntPrange}) inspired by the classical attack of Prange. 
Note that Prange's algorithm can be described as choosing $k+1$ columns  in $\begin{bmatrix} \G \\ \r \end{bmatrix}$ where $\r$ is the received word and checking whether the $(k+1) \times (k+1)$ submatrix formed at these positions is rank deficient.
This formulation can neatly be generalized to interleaved codes, where we pick $k+\ell$ columns in $\begin{bmatrix} \G \\ \R \end{bmatrix}$ and check if the rank of the $(k+\ell) \times (k+\ell)$ submatrix formed at these positions is less than $k+\ell$.

 In more details, we choose  a set $\cJ \subset [1,n]$ of size $k+\ell$, which contains an information set $\cI$ for $\G$ (in other words, $\G_\cI$ and hence $\G_\cJ$ has full rank). Let us denote again by $\G' \coloneqq \begin{bmatrix} \G \\ \R \end{bmatrix}$ and check if the square submatrix $\G'_{\cJ}$ (the blue region in Fig.~\ref{fig:intPrange}) is rank-deficient, that is $$\text{rk}\left( \left(\G' \right)_\mathcal{J}\right) <k+\ell.$$ This can be split into two cases:
 \begin{enumerate} 
 \item $\E_\cJ$ has linearly dependent rows (which implies $\G'_\cJ$ is rank deficient).
 \item $\E_\cJ$ has linearly independent rows but $\G'_\cJ$ is still rank deficient.
 \end{enumerate}
 In the first case, we succeed as at least one non-zero word in $\langle{\R} \rangle$ is error-free at these $k+\ell$ positions and so by performing the re-encoding step of Prange on such a word $\r$, we can find the error in this word, giving us a subset of the $t$ error positions.
 A naive way to find such an $\r$ would be to do the re-encoding on all $q^\ell - 1$ non-zero words in $\langle{\R} \rangle$.
 However, this step will fail if we are in the second case.
 As it turns out, the second case is far more likely than the first, so this naive re-encoding approach will make the entire algorithm very inefficient.
 Instead we do the re-encoding for only those $\r \in \langle{\R} \rangle$ that actually belong to some linearly dependent set of rows in $\G'_\cJ$
 which can be easily found by computing its left null space, i.e., the set $\{\x \in \Fq^{k + \ell} : \x\G'_\cJ = \0\}$.
 With this modification, the algorithm becomes efficient again, though perhaps at the expense of a more involved complexity analysis.

 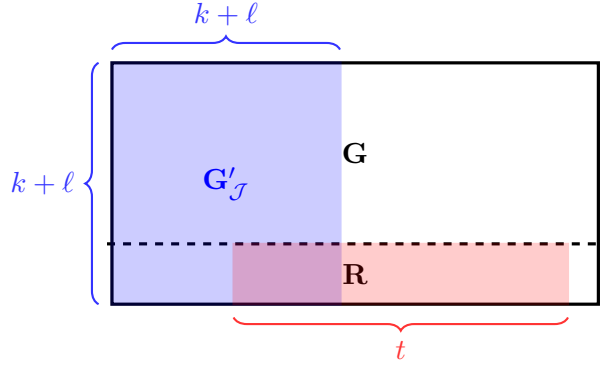
\begin{figure}[ht]
  \centering
\def\x{0.4} 

\begin{tikzpicture}
\pgfplotsset{compat = 1.3}
\node (LU) at (0,0) [draw=none] {};
\node (RB) at ($(LU)+(\x*16, \x*8)$) [draw = none] {};

\node (EL) at ($(LU)+(\x*-0.5, \x*2)$) [draw=none] {};
\node (ER) at ($(LU)+(\x*16.5, \x*2)$) [draw=none] {};

\draw[dashed, very thick] (EL) -- (ER);

\draw[very thick] (LU) rectangle (RB);
\node (E) at ($(LU)+(\x*8, \x*1)$) [draw=none] {$\R$};
\node (G) at ($(LU)+(\x*8, \x*5)$) [draw=none] {$\G$};

\node (WL) at ($(LU)+(\x*0, \x*0)$) [draw=none] {};
\node (WR) at ($(LU)+(\x*7.5, \x*8)$) [draw=none] {};
\draw[very thick, blue, fill=blue, opacity=0.2] (WL) rectangle (WR) node[pos=.5, opacity=1] {$\G_{\cJ}'$};
\draw [blue, opacity=0.8, thick,
    decorate,
    decoration = {brace,raise=5pt,
        amplitude=5pt}] ($(WL)+(\x*0,\x*8)$) -- ($(WR)+(0,0)$) node[pos=0.5,above=10pt,blue]{$k+\ell$};
\draw [blue, opacity=0.8, thick,
    decorate,
    decoration = {brace,raise=5pt,
        amplitude=5pt}] ($(WL)+(\x*0,\x*0)$) -- ($(WR)-(\x*7.50,0)$) node[pos=0.5,left=10pt,blue]{$k+\ell$};

\node (EWL) at ($(LU)+(\x*4, \x*0)$) {};
\node (EWR) at ($(LU)+(\x*15, \x*2)$) {};
\draw[thick, red, fill=red, opacity=0.2] (EWL) rectangle (EWR);
\draw [red, opacity=0.8, thick,
    decorate,
    decoration = {brace,raise=5pt,
        amplitude=5pt,mirror}] ($(EWL)+(\x*0,\x*0)$) --  ($(EWR)-(\x*0,\x*2)$) node[pos=0.5,below=10pt,red]{$t$};
\end{tikzpicture}

  \caption{Illustration of Interleaved Prange Algorithm. }
  \label{fig:intPrange}
\end{figure}

\begin{algorithm}
  \caption{Interleaved Prange}\label{algo:IntPrange}
  \SetAlgoLined
  \DontPrintSemicolon
  \KwIn{A generator matrix $\G \in\Fq^{k \times n}$ and a received matrix $\R = \C + \E \in \Fq^{\ell \times n}$ where $\E$ has at most $t$ non-zero columns.} 
  \KwOut{A nonempty subset $\cU \subseteq \supp(\E)$}

  Choose  $\cJ\subset[1,n]$ of size $k+\ell$ such that $\rank(\G_\cJ) = k$
  \label{line:chooseI}\;
  \uIf{$\rank(\G'_{\cJ})<k+\ell$
\label{line:rankcheck}
  }
  {
    \For{each $\x \in \Fq^{k+\ell} \setminus \{\mathbf{0}\}$ in the left null space of $\G'_\cJ$
    \label{line:foreachloop}
    }{
     \lIf{$\wtH(\x\G') \leq t$}{return $\supp(\x\G')$.}
   }}\uElse{Go back to step \ref{line:chooseI}.}
\end{algorithm}

\begin{theorem}
The cost of Interleaved Prange on an $\ell$-interleaved random code over $\mathbb{F}_q$ with length $n$ and dimension $k$ is in
$$\mathcal{O}\left(P^{-1} C  \right),$$ where 
\begin{align*}
    P &= \sum_{i=0}^{\min\{t, k + \ell\}}
  \dfrac{\binom{n-t}{k+\ell-i}\binom{t}{i}}{\binom{n}{k+\ell}}
  \cdot \left(1 - \prod_{j=0}^{\ell - 1} (1 - q^{j-i})\right), \\
\end{align*} 
denotes the success probability and
$$C = ( k+\ell)^3 + \prod_{j=0}^{k-1}(1-q^{j-k})16 \sum_{p=1}^\ell q^{ -p^2+p} )(k+\ell)(n-k-\ell)$$
denotes the cost of one iteration.
\end{theorem}

\begin{proof}
This algorithm succeeds whenever the chosen set $\mathcal{J}$ is such that the rows of $\E_{\cJ \cap \cT}$ are linearly dependent and that $\G_\cJ$ has rank $k$. Since the latter is true with high probability, we will assume this probability is one.
Let $i$ denote $|\cJ \cap \cT|$, i.e., the number of error positions in the set $\cJ$.
Since $\E_{\cJ \cap \cT}$ has the distribution of a random matrix in $\Fq^{\ell \times i}$,
the probability that $\E_{\cJ \cap \cT}$ has linearly dependent rows is given by,
$$ \left(1 - \prod_{j=0}^{\ell - 1} (1 - q^{j-i})\right).$$
Next, we weight this term with the probability that exactly $i$ errors land in $\cJ$ and form the summation over all possible $i$, giving us 
$$P = \sum_{i=0}^{\min\{t, k + \ell\}}
  \dfrac{\binom{n-t}{k+\ell-i}\binom{t}{i}}{\binom{n}{k+\ell}}
  \cdot \left(1 - \prod_{j=0}^{\ell - 1} (1 - q^{j-i})\right).$$
  Hence, we will need $P^{-1}$ many iterations until we succeed.
  Among these non-successful iterations, we could either have that $\G'_\mathcal{J}$ was not rank deficient, which costs 
  $\mathcal{O}((k + \ell)^3)$  due to the Gaussian elimination to check $\G'_\cJ$'s rank (step \ref{line:rankcheck})
  or we have that $\G'_\mathcal{J}$ was indeed rank deficient but $\E_\mathcal{J}$ had linearly independent rows.
  In this second case
  we incur the additional cost of step \ref{line:foreachloop} since only after that we will recognize
  that $\E_\mathcal{J}$ did not have linearly dependent rows.
  Note that step $\ref{line:foreachloop}$ consists of computing the left null space of $\G'$ and then performing $q^p$ re-encoding steps where $p$ is the dimension of this space.
  The left null space can be found using Gaussian elimination and thus has the same cost as the rank-check in step \ref{line:rankcheck}
  (in fact, it is possible to find the null space with effectively no additional work from this step).

  In order to compute the cost of doing the $q^p$ re-encodings, assume $\E_\cJ$ has linearly independent rows. 
  Let $P(p)$ denote the probability that $\G'_\mathcal{J}$ has rank deficiency $p$,
  i.e., $\text{rk}(\G'_\mathcal{J})=k+\ell-p$ where $p \in [1, \ell]$ (since $\cJ$ is chosen such that $\text{rk}(\G_\cJ) = k$). Thus $p$ is the dimension of the left null space of $\G'$. 
Then the workfactor of the re-encoding is given by
$$C' = \sum_{p=1}^\ell P(p)q^p \alpha$$
where $\alpha \in \mathcal{O}((k+\ell)(n-k-\ell))$ is the cost of a single re-encoding step.

To compute $P(p)$, we make use of the following result: 
if $V$ is an $n$\nobreakdash-dimensional vector space over $\mathbb{F}_q$ and $U$ is an $m$-dimensional subspace in $V$, then the number of $k$-dimensional subspaces $W$ over $\mathbb{F}_q$ with $\text{dim}(W \cap U) =d$ is given by
$$\gb{n-m}{k-d}_q\gb{m}{d}_q q^{(m-d)(k-d)},$$
where $\gb{a}{b}_q$ denotes the Gaussian binomial coefficient.

Since the number of $\G'_\mathcal{J}$ with rank deficiency $p$ is given by the number of $\E_{\cJ \cap \cT}$ of rank~$\ell$ times the number of $\G_\mathcal{J}$ of rank $k$ such that $\text{dim}(\langle \E_{\cJ \cap \cT} \rangle \cap \langle \G_\mathcal{J}\rangle)=p$, we get
$$\prod_{j=0}^{\ell-1}(q^i-q^j)\prod_{j=0}^{k-1}(q^k-q^j)\gb{\ell}{p}_q \gb{k}{k-p}_q q^{(\ell-p)(k-p)}$$
where the first term counts the number
of rank $\ell$  matrices $\E_{\cJ \cap \cT}$, the second term is the number of ways of picking an ordered basis of a $k$-dimensional subspace
and the third term counts the number of $k$-dimensional subspaces (i.e. $\langle \G_\cJ \rangle$) 
inside a $k + \ell$ dimensional space
whose intersection with a fixed $\ell$-dimensional subspace (i.e. $\langle \E_{\cJ \cap \cT} \rangle) $ has dimension $p$.

Dividing this by the total number of possible $\G_\mathcal{J}'$, i.e.,
$$q^{(k+\ell)k} \prod_{j=0}^{\ell-1}(q^i-q^j)$$ we get the probability 
\begin{align*} P(p)&= \prod_{j=0}^{k-1} \frac{q^k-q^j}{q^{k+\ell}}\gb{\ell}{p}_q\gb{k}{k-p}_q q^{(\ell-p)(k-p)}.
\end{align*}
Hence the workfactor of one iteration is given by 
$C = \beta + C'$ where $\beta \in \mathcal{O}((k + \ell)^3)$ and

\begin{align*}
    C' &= \sum_{p=1}^\ell P(p)q^p \alpha \\
     & = \prod_{j=0}^{k-1}\left(1-q^{j-k}\right)q^{-\ell k}\sum_{p=0}^\ell \gb{\ell}{p}_q\gb{k}{k-p}_q q^{(\ell-p)(k-p)} q^p \alpha\\
     & \leq  \prod_{j=0}^{k-1}(1-q^{j-k})q^{-\ell k}\sum_{p=1}^\ell 16 q^{\ell k -p^2+p} \alpha
     \\ 
     &= \prod_{j=0}^{k-1}(1-q^{j-k})16 \sum_{p=1}^\ell q^{ -p^2+p} \alpha,
\end{align*}
where we used that $$q^{(a-b)b} \leq \gb{a}{b}_q \leq 4 q^{(a-b)b}.$$ 
 
\end{proof}
Again, we will give an upper bound on the  asymptotic cost. For this it is enough to consider a lower bound on the success probability $P$, as \begin{align*}
   \lim_{n \to \infty} \frac{1}{n}  \log_q(C) & =   \lim_{n \to \infty} \frac{1}{n}  \log_q\left((k+\ell)^3 + \prod_{j=0}^{k-1}(1-q^{j-k})16 \sum_{p=1}^\ell q^{ -p^2+p} \alpha \right) \\
    & \leq  \lim_{n \to \infty} \frac{1}{n}  \log_q((k+\ell)^3+\ell 16 (k+\ell)(n - k - \ell) )=0.
\end{align*}
Note that the success probability can be written as 
$$P=\sum_{i=0}^{\min\{t, k+\ell\}} Q_i,$$ for $$Q_i =  \dfrac{\binom{n-t}{k+\ell-i}\binom{t}{i}}{\binom{n}{k+\ell}}
  \cdot \left(1 - \prod_{j=0}^{\ell - 1} (1 - q^{j-i})\right).$$
  To get a lower bound, we use that $\sum_{i=0}^{\min\{t, k+\ell\}} Q_i \geq Q_{\ell-1}$, that is we just consider $$ Q_{\ell-1} =  \binom{n-t}{k+1}\binom{t}{\ell-1}\binom{n}{k+\ell}^{-1}.$$
    Since the interleaving order $\ell$ is usually very small compared to $n$, we set 
    $$L= \lim_{n \to \infty} \frac{\ell(n)}{n} = \frac{T}{ \gamma}, $$ for some positive integer $2 <  \gamma.$
  \begin{proposition}
  The asymptotic complexity of Interleaved Prange on  an $\ell$-interleaved random code over $\mathbb{F}_q$ with length $n$ and dimension $k$ is given by at most $q^{ne(R,q)}$, where 
  \begin{align*}
      e(R,q) &= H(1,R+L) - H(1-T,R)-H(T,L)+ \min\{H(R+L,R),L\}.
  \end{align*}
  \end{proposition}

\subsection{Comparison}\label{sec:comp}
In general, for small $\ell$, it would appear that CF-based algorithms have a lower complexity than SD-based algorithms because SD-based algorithms generally solve the problem for a larger error weight than CF-based ones (but only in a slightly larger code). 

 However, this comparison does not take into account the  large failure probabilities of these algorithms. For this reason we will not compare the cost of these algorithms with the SD-based and the Interleaved Prange algorithm. 
 In order to compare the different algorithms, we fix $q=7$ and $\ell$ the interleaving order to be such that $L= \lim_{n \to \infty} \frac{\ell(n)}{n} = T/20,$ i.e., $ \gamma=20.$
 In addition, we denote by $R^* = \text{argmax}_{0 \leq R \leq 1} \left( e(R,q) \right)$.
 We have two different approaches for the comparison. The first one is to take $\frac{1}{n}\log_q(\cdot)$ of the actual cost of the algorithms computed for large $n$, which seem to converge rather quickly, thus Figure \ref{plot direct}  and Table \ref{tabdir} gives a very accurate plot of the complexities in this case.
 
\begin{figure}[h!]
\begin{center}
    \includegraphics[width=\textwidth]{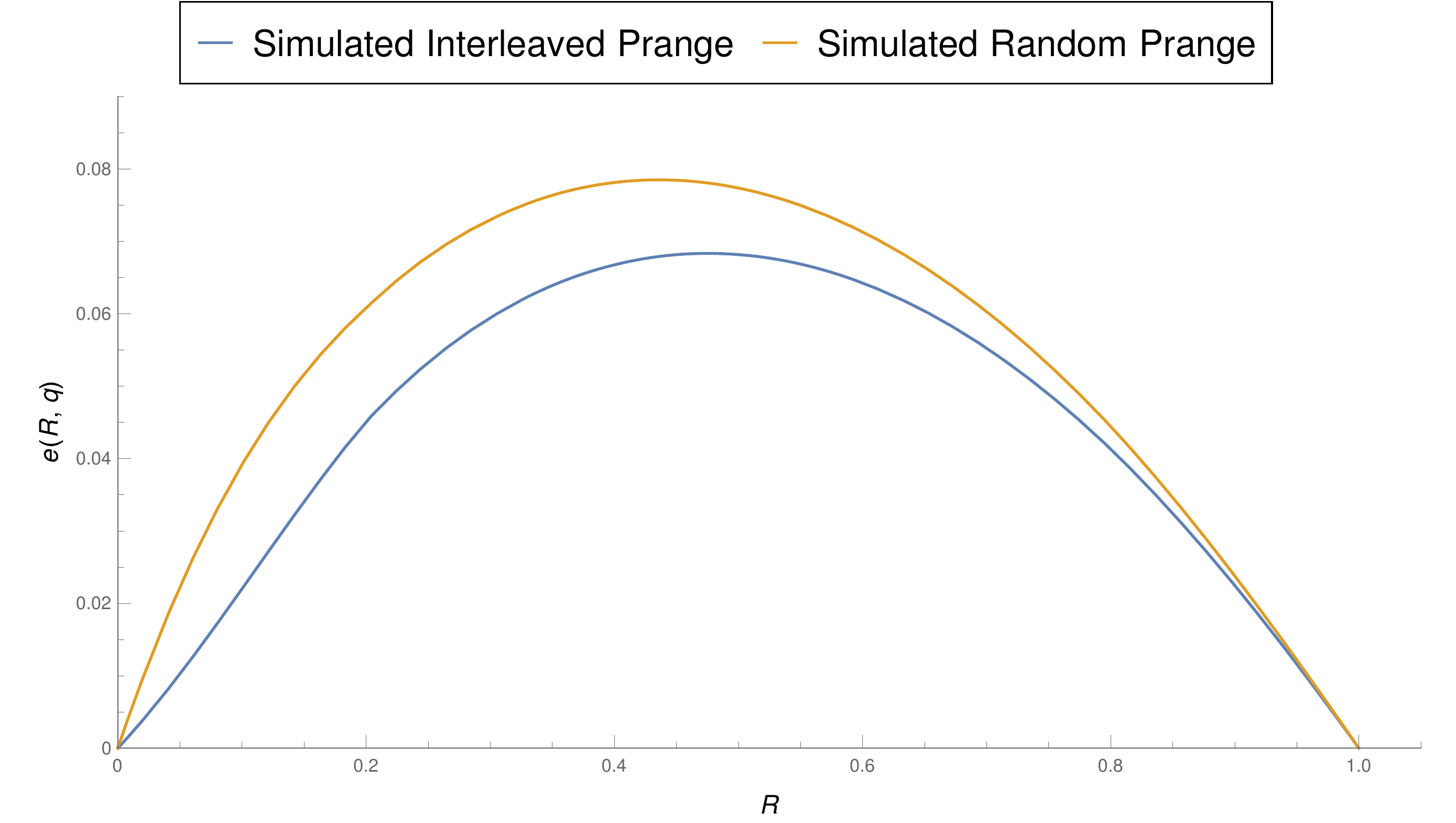}
\end{center}
\caption{Simulated asymptotic cost of the algorithms for $q=7$.}\label{plot direct}
\end{figure}

\renewcommand{\arraystretch}{1.5}
\begin{table}[h!]
 \begin{center}
 \begin{tabular}{|c|c|c|}
 \hline 
  Algorithm & $~~~e(R^*,q)~~~$ & $R^*$ \\\hline
 Simulated Interleaved Prange & 0.06832 & 0.475 \\
 Simulated Random Prange & 0.07848 & 0.437 \\
  \hline
  \end{tabular}
\end{center}  
  \caption{Comparison of simulated asymptotic cost of different algorithms for $q=7$.}\label{tabdir}
\end{table}

The second approach is using the presented upper bounds on the asymptotic complexity, which can be seen in Figure \ref{plotup} and Table \ref{tabup}.
\begin{figure}[h!]
\begin{center}
    \includegraphics[width=\textwidth]{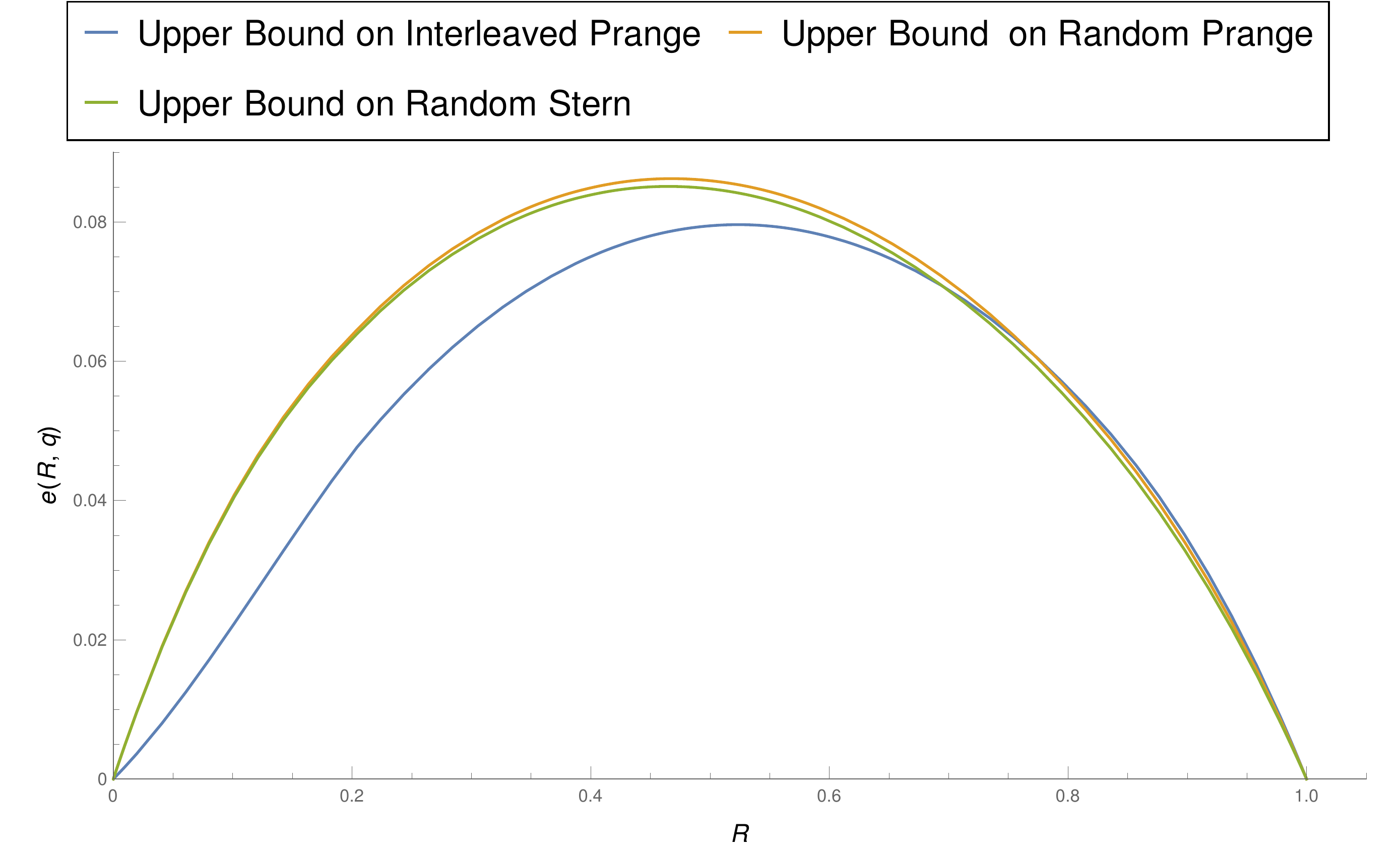}
\end{center}
\caption{Upper bounds on the asymptotic cost of the algorithms for $q=7$.}\label{plotup}
\end{figure}

\renewcommand{\arraystretch}{1.5}
\begin{table}[h!]
 \begin{center}
 \begin{tabular}{|c|c|c|}
 \hline 
  Algorithm & $~~~e(R^*,q)~~~$ & $R^*$ \\\hline
Upper Bound   Interleaved Prange & 0.07961 & 0.524 \\
  Upper Bound Random Prange & 0.08621 & 0.468 \\
 Upper Bound  Random Stern &0.08510 & 0.465\\
  \hline
  \end{tabular}
\end{center}  
  \caption{Comparison of upper bounds of asymptotic cost of different algorithms for $q=7$.}\label{tabup}
\end{table}

 Both approaches show the same predicted behaviour, that is Interleaved Prange has a much lower complexity than the straightforward approach. Note that in the simulated asymptotics we did not compare also to Random Stern, as the improvement on Random Prange is only marginal.

\section{Conclusion}\label{sec:concl}
In this paper we presented several algorithms that decode a random homogeneous $\ell$-interleaved code, which also work for the missing case $\ell \ll t$. Two of these algorithms come from a straight-forward reduction to known ISD and CF algorithms in the classical case. In addition to those algorithms, we also presented a new generic decoding algorithm for interleaved codes, namely Interleaved Prange, which is an adaption of Prange's classical idea to the interleaved setting. We provided a complexity analysis and compared the asymptotic costs of the considered algorithms. 

\subsubsection*{Acknowledgements}
The sixth author  is  supported by the Swiss National Science Foundation grant number 195290.

\bibliographystyle{plain}
\bibliography{main}
\end{document}